\documentclass[pra,aps,showpacs,amsfonts,twocolumn,nofootinbib]{revtex4-2}
\usepackage{color}
\definecolor{myurlcolor}{rgb}{0, 0 ,0.4}
\definecolor{mycitecolor}{rgb}{0, 0.5 ,0}
\definecolor{myrefcolor}{rgb}{0.5 ,0 ,0}
\usepackage[colorlinks=true, allcolors=blue]{hyperref}
\usepackage[mathscr]{euscript}
\usepackage[draft]{fixme}
\usepackage{tikz}
\usepackage{subfigure}
\usepackage{graphicx}
\usepackage{amsfonts}
\usepackage{amsmath, amssymb, amsthm,verbatim, graphicx,bbm}
\usepackage[commandnameprefix=ifneeded]{changes}
\usepackage{palatino}
\usepackage{mathpazo}
\usepackage{dsfont}
\usetikzlibrary{patterns}
\usepackage{mathtools}

\newcommand{\beq}[0]{\begin{equation}}
\newcommand{\eeq}[0]{\end{equation}}

\newcommand{\pz}[0]{p(\boldsymbol{0}|\boldsymbol{0})}

\newcommand{\ket}[1]{|#1\rangle}
\newcommand{\bra}[1]{\langle#1|}

\newcommand{\one}{\leavevmode\hbox{\small1\normalsize\kern-.33em1}}

\def\be{\begin{equation}}
\def\ee{\end{equation}}
\def\ben{\begin{eqnarray}}
\def\een{\end{eqnarray}}
\def\eea{\end{array}}
\def\bea{\begin{array}}

\newcommand{\bei}{\begin{itemize}}
\newcommand{\eei}{\end{itemize}}

\newcommand{\Ke}[1]{\big|#1\big\rangle}
\newcommand{\Br}[1]{\big< #1\big|}

\renewcommand{\emph}[1]{\textbf{#1}}

\theoremstyle{plain}
\newtheorem{thm}{Theorem}

\theoremstyle{definition}

\theoremstyle{remark}

\newcommand{\orcid}[1]{\href{https://orcid.org/#1}{\includegraphics[width=7pt]{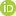}}}

\begin{document}
\title{Single Bell inequality to detect genuine nonlocality in three-qubit pure genuinely entangled states}

\author{Ignacy Stachura}
\author{Owidiusz Makuta \orcid{0000-0002-0070-8709}}
\author{Remigiusz Augusiak \orcid{0000-0003-1154-6132}}
\email{augusiak@cft.edu.pl}

\affiliation{Center for Theoretical Physics, Polish Academy of Sciences, Aleja Lotnik\'ow 32/46, 02-668 Warsaw, Poland}

\begin{abstract}
    It remains an open question whether every pure multipartite state that is genuinely entangled is also genuinely nonlocal. Recently, a new general construction of Bell inequalities allowing the detection of genuine multipartite nonlocality (GMNL) in quantum states was proposed in [F. J. Curchod, M. L. Almeida, and A. Ac\'in, \href{https://iopscience.iop.org/article/10.1088/1367-2630/aaff2d}{New
    J. Phys. \textbf{21}, 023016 (2019)}] with the aim of addressing the above problem. Here we show how, in a simple manner, one can improve this construction to deliver finer Bell inequalities for detection of GMNL. Remarkably, we then prove one of the improved Bell inequalities to be powerful enough to detect GMNL in every three-qubit genuinely entangled state. We also generalize some of these inequalities to detect not only GMNL but also nonlocality depth in multipartite states and we present a possible way of generalizing them to the case of more outcomes.
\end{abstract}

\maketitle
\section{Introduction}

Using entangled quantum particles one may perform an experiment which gives rise to correlations that can by no means be explained by a local realistic theory. The existence of such correlations, referred to as Bell nonlocality, is one of the most characteristic features of quantum mechanics. Moreover, nonlocal correlations have been harnessed as a resource for numerous device-independent applications such as quantum cryptography \cite{Ac_n_2007}, certification of randomness \cite{Pironio_2010} or self-testing \cite{10.5555/2011827.2011830,_upi__2020}.

One of the fundamental problems in quantum information and foundations of quantum physics is understanding the relationship between the Bell nonlocality and quantum entanglement, another key feature of quantum theory. While the former implies the latter in the sense that a state that exhibits Bell non-locality must be necessarily entangled, the opposite implication does not always hold true. Even though a considerable effort has been devoted to 
understanding which entangled states do or do not violate Bell inequalities (see, e.g., Refs. \cite{PhysRevA.40.4277,PhysRevLett.115.030404,Augusiak_2014,Bowles_2016}), the relationship between entanglement and nonlocality for general composite quantum systems remains unclear even in case of pure states. In fact, while this connection has been elucidated for two-particle systems, the exploration of the interplay between entanglement and nonlocality in systems involving more than two particles remains an active area of investigation. Indeed, it was demonstrated in Refs. \cite{GISIN1991201,Cabello_2002} that every pure entangled state violates local realism. Subsequently, in Ref. \cite{POPESCU1992293}, Popescu and Rohrlich generalized this observation to the multipartite setting proving that any pure entangled state displays some form of nonlocality.

While one could conclude that these results definitely resolve the problem, this is not entirely the case because the truly relevant question in the multipartite setting is whether genuine entanglement \cite{PhysRevA.65.012107} implies genuine nonlocality \cite{PhysRevD.35.3066, PhysRevA.88.014102}. There have been attempts to settle this problem. In particular, in Ref. \cite{yu2013tripartite}, using a three-partite form of Hardy-type argument, it was proven that any three-qubit genuinely entangled state is also genuinely nonlocal, whereas in Ref. \cite{Chen_2014} the same conclusion was drawn for arbitrary $n$-qubit symmetric states. Despite this progress, the question of whether any genuinely entangled state is genuinely nonlocal, beyond the above simple cases, remains open. 

A possible way to tackle this problem goes through the construction of suitable Bell inequalities whose violation reveals GMNL. While many such Bell inequalities have been derived to date (see, e.g., Refs. \cite{PhysRevD.35.3066,PhysRevA.88.014102,Bancal_2009,PhysRevLett.108.100401,Augusiak_2019,Pandit_2022}), none of them has been proven to be powerful enough to demonstrate genuine nonlocality in every genuinely entangled pure multipartite state.

Another technique to construct Bell inequalities witnessing GMNL which seems promising in this context was proposed recently by Curchod and collaborators in Ref. \cite{Curchod_2019}. It takes a bipartite Bell inequality, called a seed inequality, that obeys certain requirements and reforges it into a multipartite one detecting GMNL. Importantly, it was analytically proven in Ref. \cite{Curchod_2019} that the inequalities obtained in this way with the seed inequality being the CHSH Bell inequality \cite{PhysRevLett.23.880} are capable of revealing GMNL of any three-qubit pure state which is symmetric under permutation of two chosen parties as well as of GHZ-like states of any number of qubits. These analytical results were also supplemented by numerical evidence that one of these inequalities is violated by any four-qubit GME state. It is worth mentioning that more recently, suitable modifications of these Bell inequalities were used in Ref. \cite{Contreras_Tejada_2021} to characterize GMNL in quantum networks. 

Our aim in this work is to show that a slightly improved version of one of the inequalities of Ref. \cite{Curchod_2019} is powerful enough to definitely settle the problem of equivalence of genuine nonlocality and genuine entanglement for three-qubit states. In other words, we introduce a single Bell inequality that is capable of detecting genuine nonlocality in any three-qubit genuinely entangled. We thus not only recover and confirm the result of Ref. \cite{yu2013tripartite}, but we also introduce a universal test of genuine nonlocality for three-qubit states which, unlike the Hardy-type arguments which require a certain set of equalities to be satisfied, is capable of tolerating some amount of experimental imperfections and noises. On the way of achieving the above result we show how to modify the construction of Ref. \cite{Curchod_2019} to provide finer Bell inequalities detecting GMNL which improves that of Ref. \cite{Curchod_2019}. We also generalize some of the inequalities of Ref. \cite{Curchod_2019} to detect the nonlocality depth of multipartite systems. Finally, we explore the possibility of generalizing the approach of Ref. \cite{Curchod_2019} to Bell scenarios involving more outcomes, concentrating on a particular case of three-outcome measurements. Indeed, we show how to suitably modify the CGLMP Bell inequality \cite{Collins_2002} to derive an inequality detecting GMNL in tripartite systems which is analogous to inequalities of Ref. \cite{Curchod_2019}

\section{Preliminaries}

Before presenting our main results, let us first introduce some relevant terminology. We will introduce the notions of genuine multipartite entanglement (GME) and genuine multipartite nonlocality (GMNL) and then in order to make the paper self-contained we will recall  the constructions of Bell inequalities 
detecting GMNL from Ref. \cite{Curchod_2019}.

\subsection{Genuinely multipartite entanglement and nonlocality}

Let us consider $n$ observers $A_1,\ldots,A_n$ sharing some $n$-partite pure state $\ket{\psi}$ from the Hilbert space $\mathcal{H}=\mathbb{C}^{d_1}\otimes\ldots\otimes \mathbb{C}^{d_n}$. Let us then consider a partition of all the parties into two disjoint and nontrivial sets $S$ and $\bar{S}$ such that $S\cup \bar{S}=\mathcal{A}:=\{A_1,\ldots,A_n\}$. We say that $\ket{\psi}$ is entangled across this bipartition if $\ket{\psi}\neq \ket{\phi_S}\otimes\ket{\phi_{\bar{S}}}$ for any two vectors $\ket{\phi_S}$ and $\ket{\phi_{\bar{S}}}$ corresponding to the parties in $S$ and $\bar{S}$, respectively. If the state $\ket{\psi}$ is entangled with respect to all bipartitions, we call it genuinely multipartite entangled (GME) \cite{PhysRevA.65.012107}. An excellent example of a pure GME state is the $n$-qudit Greenberger-Horne-Zeilinger (GHZ) state \cite{Greenberger1989}
\begin{equation}\label{GHZ}
    \ket{\mathrm{GHZ}_{n,d}}:=\frac{1}{\sqrt{d}}\sum_{i=0}^{d-1}\ket{i}^{\otimes n}.
\end{equation}

The above definition of entanglement can further be extended to mixed states. Namely, a state $\rho$ is entangled across a bipartition $S|\bar{S}$ if it does not admit the following representation
\begin{equation}
    \rho=\sum_{\lambda}q_{\lambda}\,\rho_{S}^{\lambda}\otimes\rho_{\bar{S}}^{\lambda},
\end{equation}
where $\rho_{S}^{\lambda}$ and $\rho_{\bar{S}}^{\lambda}$ are quantum states corresponding to 
the sets of parties $S$ and $\bar{S}$, whereas $q_{\lambda}$ is some probability distribution. 
A state that admits the above representation is called separable across $S|\bar{S}$.
We then say that $\rho$ is GME \cite{PhysRevA.65.012107} if it cannot be represented as a convex combination of states that are separable across any possible bipartition, 
\begin{equation}
\rho=\sum_{S|\bar{S}}r_{S|\bar{S}}\sum_{\lambda}q_{\lambda}^{S|\bar{S}}\,\rho_{S}^{\lambda}\otimes\rho_{\bar{S}}^{\lambda},
\end{equation}
where $r_{S|\bar{S}}$ is some probability distribution and the sum is over all nontrivial bipartitions for which both sets $S$ and $\bar{S}$ are nonempty.

Let us now move on to the notions of genuine multipartite nonlocality and nonlocality depth. We again consider $n$ parties sharing a state $\rho$ and we assume that each party performs one of $m$, $d$-outcome measurements on their share of the state. We label the measurements choices and outcomes of party $A_i$ as $x_i = 1,\ldots,m$ and $a_i=0,\ldots,d-1$, respectively. After many rounds of repeated measurements, the parties can estimate the underlying joint probabilities $p(a_1,\ldots,a_n|x_1,\ldots,x_n)$ of obtaining results $a_1,\ldots,a_n$ upon performing the measurements labeled by $x_1,\ldots,x_n$.

The correlations that are produced in this experiment are described by a collection of probability distributions that can be ordered into a vector
\begin{equation}
    \vec{p}:=\{p(\vec{a}|\vec{x}):=p(a_1,\ldots,a_n|x_1,\ldots,x_n)\},
\end{equation}
where $\vec{a}$ and $\vec{x}$ are shorthands for $a_1,\ldots,a_n$ and $x_1,\ldots,x_n$, respectively.

To consider correlations as local in a given scenario, we require that the outcome obtained by one party depends neither on the outcomes nor on the measurements of the other parties. It is still allowed that the probabilities can be correlated by some shared information represented by the variable $\lambda$, and thus we define local correlations as those that admit the following decomposition:
\begin{equation}
\label{local}
    p(\vec{a}|\vec{x})=\sum_{\lambda}q_{\lambda}p_{A_1}(a_1|x_1,\lambda)\ldots p_{A_n}(a_n|x_n,\lambda),
\end{equation}
where $q_{\lambda}$ is some probability distribution and $p_{A_i}(a_i|x_i,\lambda)$ are probability distributions representing the action of each party $A_i$. Correlations that cannot be written in the form \eqref{local} are termed Bell nonlocal. Then, the state $\rho$ is called nonlocal if it gives rise to nonlocal correlations in some Bell scenario.

A more suitable notion in the multipartite case is that of genuine nonlocality. Consider again a bipartition $S|\bar{S}$ with nonempty sets $S$ and $\bar{S}$. We call a correlation $\vec{p}$ local across the bipartition $S|\bar{S}$ if
\begin{equation}\label{eq:local_across}
p(\boldsymbol{a}|\boldsymbol{x})=\sum_{\lambda}q_{\lambda}p_{S}(\boldsymbol{a}_S|\boldsymbol{x}_S,\lambda)p_{\bar{S}}(\boldsymbol{a}_{\bar{S}}|\boldsymbol{x}_{\bar{S}},\lambda),
\end{equation}
where $\boldsymbol{a}_{S}$ and $\boldsymbol{x}_{S}$ are vectors of outcomes and measurement choices of parties belonging to the set $S$, and $p_{S}(\boldsymbol{a}_{S}|\boldsymbol{x}_{S},\lambda)$ are probability distributions that satisfy the non-signaling principle \cite{Barrett_2005}. Then we say that $\vec{p}$ is bilocal if it can be represented as 
\begin{equation}\label{bilocal}
\begin{aligned}
p(\boldsymbol{a}|\boldsymbol{x})&=\sum_{S|\bar{S}}r_{S|\bar{S}}p_{S|\bar{S}}(\boldsymbol{a}|\boldsymbol{x}),
\end{aligned}
\end{equation}
where $p_{S|\bar{S}}(\boldsymbol{a}|\boldsymbol{x})$ are probability distributions local across $S|\bar{S}$ \eqref{eq:local_across}. Finally, we refer to any $\vec{p}$ that does not admit the above decomposition as genuinely multipartite nonlocal (GMNL) \cite{PhysRevD.35.3066, PhysRevA.88.014102}.

Interestingly, one can introduce a more fine-grained classification of nonlocality in the multipartite case to more accurately describe different forms of nonlocality in correlations that are neither local nor GMNL. Notions that suit this purpose are those of $k$-producibility of nonlocality and nonlocality depth. 

Let us consider a partition of the set $\mathcal{A}=\{A_1,\ldots,A_n\}$ into $L$ disjoint nonempty subsets $S_i$ such that $S_1\cup \ldots\cup S_L=\mathcal{A}$ and the size of each of them does not exceed $k$, i.e., $|S_i|\leqslant k$ for some positive integer $k=1,\ldots,n$. We refer to this partition as an $L_k$ partition of $\mathcal{A}$. Consider now correlations that admit the following decomposition:
\begin{equation}\label{Lkpart}
    p(\boldsymbol{a}|\boldsymbol{x})=\sum_{\lambda}q_{\lambda} p_{S_{1}}(\boldsymbol{a}_{S_1}|\boldsymbol{x}_{S_1},\lambda)\cdot\ldots\cdot  p_{S_{L}}(\boldsymbol{a}_{S_L}|\boldsymbol{x}_{S_L},\lambda),
\end{equation}
where as before we assume that $p_{S_{i}}(\boldsymbol{a}_{S_i}|\boldsymbol{x}_{S_i})$ satisfy non-signaling conditions between parties in $S_i$.
We call correlations $\vec{p}$
that can be written in this way $k$-producible with respect to the given $L_k$ partition.
We then call $\vec{p}$ $k$-producible if it can be decomposed into the following convex combination 
\begin{equation}\label{kprod}
    p(\boldsymbol{a}|\boldsymbol{x})=\sum_{L_k}q_{L_k}p_{L_k}(\boldsymbol{a}|\boldsymbol{x})
\end{equation}
of correlations $p_{L_k}(\boldsymbol{a}|\boldsymbol{x})$ that are $k$-producible with respect to various $L_k$ partitions. Moreover, the minimal $k$ for which $p(\boldsymbol{a}|\boldsymbol{x})$ can be expressed in the form \eqref{kprod} is
referred to as the nonlocality depth. In other words, we say that $p(\boldsymbol{a}|\boldsymbol{x})$  has nonlocality depth $k$ if $k$ is the minimal integer $k=1,\ldots,n$ for which the decomposition \eqref{kprod} exists. Hence, for $k=N$, $p(\boldsymbol{a}|\boldsymbol{x})$ is GMNL.

\subsection{Class of Bell inequalities for detection of GMNL}
Let us now briefly recall some of the Bell inequalities constructed in Ref. \cite{Curchod_2019}. 
To this end, we first need to introduce the well-known CHSH Bell inequality \cite{PhysRevLett.23.880}, which serves as a seed inequality for the construction. 
We state it here in its quite unusual form
\begin{equation}\label{CH}
I:=p(00|00)-p(01|01)-p(10|10)-p(00|11)\leqslant 0,
\end{equation}
which can be obtained from the standard form by using the nonsignaling conditions.

Then, let us denote by $I^{i,j}_{\boldsymbol{0}|\boldsymbol{0}}$ the CHSH Bell expression between the parties $A_i$ and $A_j$ lifted to $n$ parties in such a way that the measurement choices and outcomes of the remaining $n-2$ parties are $x_i=0$ and $a_i=0$, that is, 
\begin{equation}\label{Liftings}
\begin{aligned}
I^{i,j}_{\boldsymbol{0}|\boldsymbol{0}}:=&p(\boldsymbol{0}|\boldsymbol{0})-p(0\ldots0 1_i0\ldots 0|0\ldots01_i0\ldots0)\\
&-p(0\ldots 01_j0\ldots 0|0\ldots 01_j0\ldots 0)\\
&-p(\boldsymbol{0}|0\ldots 01_i0\ldots01_j0\ldots 0),
\end{aligned}
\end{equation}
where $\boldsymbol{0}:=0\ldots0$ is a shortcut notation for $n$ consecutive zeros, and the subscripts $i,j$ refer to parties $A_i$ and $A_j$. Importantly, for correlations $\vec{p}$ that are local across any bipartition $S|\bar{S}$ \eqref{local} such that $A_i\in S$ and $A_j\in \bar{S}$, the expression \eqref{Liftings} satisfies $I^{i,j}_{\boldsymbol{0}|\boldsymbol{0}}\leqslant 0$. 

In this notation, the Bell inequalities detecting GMNL from Ref. \cite{Curchod_2019}
can be stated as 
\begin{equation}\label{Isym}
\sum_{i=1}^{n}\sum_{i<j=1}^{n} I^{i,j}_{\boldsymbol{0}|\boldsymbol{0}}\leqslant\binom{n-1}{2}p(\boldsymbol{0}|\boldsymbol{0})
\end{equation}
and 
\begin{equation}\label{I1}
\sum_{j=2}^n I^{1,j}_{\boldsymbol{0}|\boldsymbol{0}}\leqslant (n-2)p(\boldsymbol{0}|\boldsymbol{0}).
\end{equation}
Notice that the first inequality is invariant under permutation of any pair of parties, whereas second one involves the lifted CHSH Bell expressions between a fixed party $A_1$, and the remaining parties $A_i$ with $i=2,\ldots,n$. 

It is worth mentioning that in the simplest Bell scenario involving two dichotomic measurements per observer the above procedure can be phrased in a more general way \cite{Curchod_2019}. In fact, any $m$-partite Bell inequality with $m<n$ which admits the following form 
\begin{equation}
\label{seed condition}
    \pz-\sum_{\mathbf{a},\mathbf{x}\neq\mathbf{0}} \beta_{\mathbf{x}}^{\mathbf{a}} p(\mathbf{a}|\mathbf{x})\leqslant0,
\end{equation}
can be used as a seed inequality for the construction.

In Ref. \cite{Curchod_2019} a few examples of Bell inequalities in the bipartite and multipartite scenarios were presented which admit the above decomposition such as for instance 
the one found in Ref. \cite{PhysRevA.88.014102} which reads
\begin{equation}
\begin{aligned}
I^{1,2,3}_{\mathrm{tri}}:= &p(000|000)-p(010|111)-p(000|011)\\
&-p(001|001)-p(100|110)-p(010|010)\\
&-p(100|100)\leqslant 0    
\end{aligned}
\end{equation}
and detects GMNL in three-partite quantum systems. Taking this inequality as a seed inequality for the construction of \cite{Curchod_2019}, the following inequality was produced
\begin{equation}\label{tripartite}
    \sum_{i>2}I_{\mathrm{tri}}^{1,2,i}\leqslant (n-3)p(\mathbf{0}|\mathbf{0}).
\end{equation}

\section{Results}

\subsection{General construction}
\label{General}

In this section, we show a general technique of formulating inequalities that detect GMNL correlations, and in particular, we show how this technique can be used to improve the inequalities of Ref. \cite{Curchod_2019}, in particular \eqref{I1}. To this end, let us for a moment take a broader perspective and consider a general expression composed of various $m$-partite Bell expressions 
\begin{equation}\label{eq4}
\sum_{g\in G}\mathcal{I}^{g},
\end{equation}
where $G$ is an arbitrarily chosen set of subsets $g\subset \{1,\ldots,n\}$ such that $|g|=m$, and $\mathcal{I}^{g}$ are in general $m$-partite Bell expressions detecting GMNL between parties from the set $\{A_i\}_{i\in g}$ lifted to the $n$-partite scenario for $n>m$ such as for instance the expressions in \eqref{Liftings}.

Our technique requires that the expressions $\mathcal{I}^{g}$ fulfill a couple of conditions.
\begin{itemize}
\item[(i)] First, for all correlations $\vec{p}$ that are local across the bipartition $S|\bar{S}$ \eqref{eq:local_across}, we require that $\mathcal{I}^{g}\leqslant 0$ if there exists a pair $i,j\in g$ such that $A_i\in S$ and $A_j\in\bar{S}$.

\item [(ii)] To formulate the second condition, we first decompose $\mathcal{I}^{g}$ as
a difference of two expressions,
\begin{equation}\label{eq:pm_decomposition}
    \mathcal{I}^{g}=\mathcal{I}^{g}_{+}-\mathcal{I}^{g}_{-},
\end{equation}
where both $\mathcal{I}^{g}_{\pm}$ only contain probabilities with positive weights. In other words, we decompose each expression $\mathcal{I}^g$ into 
two expressions containing probabilities that enter the $\mathcal{I}^g$ with positive and negative signs, respectively. Then, the second condition on expressions $\mathcal{I}^{g}$ requires that the positive parts of every Bell expression $\mathcal{I}^{g}$ are the same, that is,
\begin{equation}\label{decomp}
    \forall_{g\neq g'\in G}\qquad \mathcal{I}^{g}_{+}=\mathcal{I}^{g'}_{+}=:\mathcal{I}_+.
\end{equation}

\end{itemize}
%

Let us also denote by $T$ the combination of probabilities that is common to all the negative parts $\mathcal{I}_{-}^{g}$. In other words, we decompose all $\mathcal{I}_{-}^{g}$ as a sum of $T$ and the remaining terms that are different for different $\mathcal{I}_{-}^{g}$; in particular, if $\mathcal{I}_{-}^{g}$ do not share any common terms, $T=0$.  
All together this implies that for all $g\in G$ we have
\begin{equation}\label{eq:bound_I^ij}
    \mathcal{I}^{g} \leqslant \mathcal{I}_{+} - T
\end{equation}

Our aim now is to determine an upper bound for the expression \eqref{eq4} for bilocal models \eqref{bilocal}. For this purpose, we consider a particular bipartition $S|\bar{S}$. Then, for all $\vec{p}$ that are local across the bipartition $S|\bar{S}$ \eqref{eq:local_across} we have
\begin{equation}
    \sum_{g\in G}\mathcal{I}^{g}\leqslant
    \sum_{g\in G'}\mathcal{I}^{g},
\end{equation}
where $G'\subseteq G$ is defined as 
%
%
\begin{equation}
    G'=\{g\in G\,|\,\{A_{i}\}_{i\in g}\subset S \,\mathrm{or}\, \{A_{i}\}_{i\in g}\subset \bar{S}\};
\end{equation}
for the remaining $g\in G$ such that $g\notin G'$ we employed the condition (i) stated above. Next, using the decomposition \eqref{decomp}, and Ineq. \eqref{eq:bound_I^ij}, the above can further be upper bounded by 
\begin{equation}\label{eq:ineq_G}
    \sum_{g\in G}\mathcal{I}^{g}\leqslant |G'|\left(\mathcal{I}_{+}-T\right).
\end{equation}
This leads us to our first result which we formulate as the following theorem.
\begin{thm}\label{thm:GMNL_ineq}
Consider a Bell expression \eqref{eq4} such that the conditions (i) and (ii) are satisfied by 
$\mathcal{I}^g$ for each $g\in G$. Then, violation of the following Bell inequality

%
\begin{equation}\label{eq:ineq_proof}
    \sum_{g\in G}\mathcal{I}^{ij}\leqslant \gamma\left(\mathcal{I}_{+}-T\right)
\end{equation}
certifies GMNL, where $\gamma=\max_{S|\bar{S}}|G'|$ and $T$ is a common term appearing in all the negative parts $\mathcal{I}_{-}^{g}$ under the decomposition \eqref{eq:pm_decomposition}. 
\end{thm}
\begin{proof}To prove the above theorem first notice that Ineq. \eqref{eq:ineq_proof} is derived from Ineq. \eqref{eq:ineq_G} by maximizing over all bipartitions $S|\bar{S}$, meaning that Ineq. \eqref{eq:ineq_proof} holds true for all $\vec{p}$ local across any $S|\bar{S}$. Next, notice that all $\vec{p}$ admitting Eq. \eqref{bilocal} are a convex combination of probability distributions $p_{S|\bar{S}}(\boldsymbol{a}|\boldsymbol{x})$ local across bipartitions $S|\bar{S}$. And since both sides of Ineq. \eqref{eq:ineq_proof} are linear functions of $\vec{p}$, the violation of Ineq. \eqref{eq:ineq_proof} implies that the correlation $\vec{p}$ does not admit Eq. \eqref{bilocal}, or in other words, that $\vec{p}$ is GMNL.
\end{proof}

\textbf{Example 1.} Let us now show how the above considerations lead to an improvement of Ineq. \eqref{I1} derived in Ref. \cite{Curchod_2019}. In this case $\mathcal{I}^{g}$ are simply the lifted CHSH Bell expressions $I^{i,j}_{\boldsymbol{0}|\boldsymbol{0}}$ \eqref{Liftings} and $G=\{\{1,i\}\}_{i=2}^{n}$ is the set of all pairs of parties $A_1A_i$ for $i>1$. Then, it directly follows that $\mathcal{I}^{+}=p(\boldsymbol{0}|\boldsymbol{0})$ and at the same time, one finds that in this case 
\begin{equation}
    T=p(1\mathbf{0}_{n-1}|1\mathbf{0}_{n-1}),
\end{equation}
where $\mathbf{0}_{n-1}$ is a shorthand for a string of $0$ of length $n-1$.
All this implies that the improved inequality is 
\begin{equation}\label{improved00}
\begin{aligned}
\sum_{j=2}^n I^{1,j}_{\boldsymbol{0}|\boldsymbol{0}}\leqslant (n-2)[p(\mathbf{0}|\mathbf{0})-p(1\mathbf{0}_{n-1}|1\mathbf{0}_{n-1})].
\end{aligned}
\end{equation}

In order to demonstrate that the above inequality is finer than that in \eqref{I1} we consider a mixture of the $n$-qubit GHZ state \eqref{GHZ} and white noise,
\begin{equation}
\label{noise}
   \rho(q):= (1-q)\ket{\mathrm{GHZ}_{n,2}}\!\bra{\mathrm{GHZ}_{n,2}}+q\frac{\mathds{1}}{2^n}
\end{equation}
with $q\in[0,1]$.
The numerical comparison between noise levels needed for violation of either inequality is presented in Fig. \ref{GHZ}. The numerical analysis shows that, in the case of a three-qubit GHZ state, the violation of Ineq. \eqref{I1} is achieved for noise levels of $q< 0.05$ while for the new inequality \eqref{improved00}  we can observe violation for $q<0.106$.
\begin{figure}[t]
\includegraphics[width=\columnwidth]{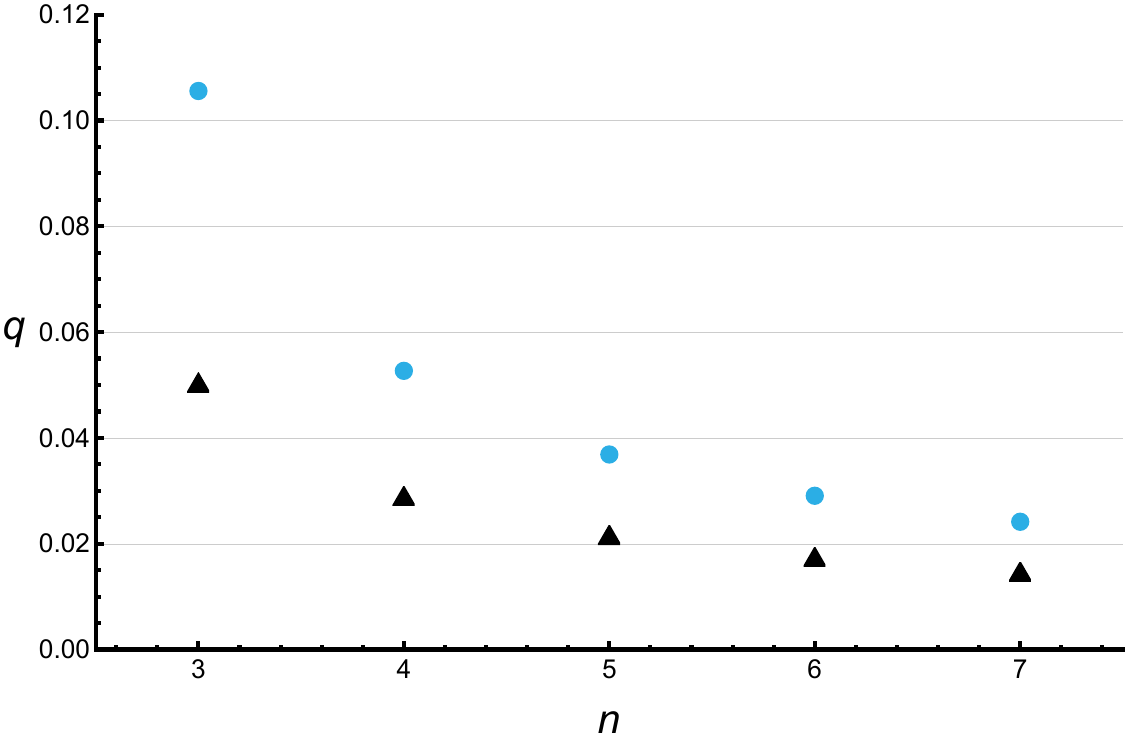}
\caption{The numerical comparison of noise robustness to the white noise of inequalities \eqref{I1} (black triangles) and \eqref{improved00} (blue disks), for a noisy GHZ state \eqref{noise}. The variable $q$ corresponds to the maximal amount of white noise for which a given inequality detects GMNL of the state $\rho(q)$ \eqref{noise}, while $n$ denotes the number of qubits. We see that even though the noise robustness of both inequalities drops with $n$, the new inequality is still outperforming the old one for every $n$.}
\label{GHZ_2}
\end{figure}
Each point in Fig. \ref{GHZ} was obtained by optimizing the Bell expressions \eqref{I1} and \eqref{improved00} over two pairs of dichotomic observables, one pair corresponding to party $A_1$ and the other pair corresponding to the remaining parties; here we assume that all parties except $A_1$ measure the same pair of observables which is motivated by the fact that the noisy GHZ state $\rho(q)$ are permutationally invariant. Each observable was parameterized by two real angles and then the standard methods in \textit{Mathematica} were used to find the maxima \cite{stachura_2024_13343476}.

\textbf{Example 2.} In a similar vein one can improve the inequality \eqref{tripartite}.
Indeed, by following the prescription above one realizes that the following inequality is satisfied by all non-GMNL correlations 
\begin{equation}
\begin{aligned}
\sum_{i>2}I_{\mathrm{tri}}^{1,2,i}&\leqslant (n-3)[p(\boldsymbol{0}|\boldsymbol{0})-p(01\boldsymbol{0}_{n-2}|01\boldsymbol{0}_{n-2})\\
&-p(10\boldsymbol{0}_{n-2}|10\boldsymbol{0}_{n-2})-p(10\boldsymbol{0}_{n-2}|11\boldsymbol{0}_{n-2})],
\end{aligned}
\end{equation}
where $\mathcal{I}_+=p(\boldsymbol{0}|\boldsymbol{0})$ and 
\begin{equation}
\begin{aligned}
T=& p(01\boldsymbol{0}_{n-2}|01\boldsymbol{0}_{n-2})-p(10\boldsymbol{0}_{n-2}|10\boldsymbol{0}_{n-2})\\
&-p(10\boldsymbol{0}_{n-2}|11\boldsymbol{0}_{n-2}).
\end{aligned}
\end{equation}

\subsection{All tripartite, pure GME states are GMNL}

Let us now show that the improved inequality \eqref{improved00} is capable of detecting genuine nonlocality in any three-qubit pure genuinely entangled. Whereas, the equivalence between genuine entanglement and genuine nonlocality for three-qubit pure states was already established with a Hardy-type argument \cite{yu2013tripartite}, here we
reproduce this result with a single Bell inequality. 

\begin{thm}\label{thm}
Every triqubit, pure, genuinely multipartite state violates Ineq. \eqref{improved00}. 
\end{thm}
While the detailed proof can be found in Appendix \ref{app:GME=GMNL}, here we outline the main idea behind it. To detect GMNL we use the improved inequality \eqref{improved00} which for $n=3$ can be rewritten as
\begin{equation}\label{eq:ineq_n=3}
\begin{aligned}
p(0&00|000)-p(100|100)-p(010|010)\\\
    &-p(001|001)-p(000|110)-p(000|101)\leqslant 0.
\end{aligned}
\end{equation}

To simplify our considerations, we make use of the generalized Schmidt decomposition for three-qubit states introduced in Ref. \cite{Carteret_2000}, which for the case $n=3$, states that every state $\ket{\psi}\in (\mathbb{C}^{2})^{\otimes 3}$ can be decomposed as
\begin{equation}\label{decompose}
\ket{\psi}=a \mathrm{e}^{\mathbbm{i}\varphi}\ket{000}+b\ket{011}+c\ket{101}+d\ket{110}+e\ket{111},
\end{equation}
where $\mathrm{e}$ is the Euler's number, $\varphi\in \mathbb{R}$, and $a,b,c,d,e$ are non-negative coefficients satisfying $a\geqslant x$ for all $x\in\{b,c,d,e\}$, and  $a^{2}+b^{2}+c^{2}+d^{2}+e^{2}=1$ . Moreover, notice that the violation of Ineq. \eqref{I1} implies the violation of Ineq. \eqref{improved00}. Therefore, since the violation of Ineq. \eqref{I1} was used to show that all pure GME states, symmetric with respect to a permutation of the second and third parties are GMNL \cite{Curchod_2019}, we can limit ourselves only to states which are not invariant under such permutation.

To show a violation of Ineq. \eqref{eq:ineq_n=3} we chose measurements parametrized by a coefficient $\alpha\in [0,\pi/2]$ such that for all $\alpha$
\begin{equation}
p(001|001)=p(010|010)=p(000|101)=0.
\end{equation}
Then, we show that for all non-symmetric states \eqref{decompose} there exists $\alpha$ for which 
\begin{equation}
\begin{aligned}
p(0&00|000)-p(100|100)-p(000|110)> 0,
\end{aligned}
\end{equation}
which implies the violation of Ineq. \eqref{eq:ineq_n=3}. 

Interestingly, since this holds for all tripartite, pure states that are not symmetric with respect to a permutation of two parties, it implies that all such states are GMNL, and so also GME.

\subsection{Detection of nonlocality depth}

Interestingly, we can easily generalize the approach of Theorem \ref{thm:GMNL_ineq} to construct inequalities detecting nonlocality depth. In fact,
the following theorem holds true.
\begin{thm}\label{thm3}
Consider a Bell expression \eqref{eq4} such that the conditions (i) and (ii) stated in Sec.  
\ref{General} are satisfied by $I^g$ for each $g\in G$. Then, violation of the following Bell inequality
\begin{equation}\label{eq:ineq_proof_k}
    \sum_{g\in G}\mathcal{I}^{g}\leqslant \gamma_{k}\left(\mathcal{I}_{+}-T\right),
\end{equation}
by correlations $p(\boldsymbol{a}|\boldsymbol{x})$,
where $\gamma_{k}=\max_{S_{1},S_{2},\ldots,S_{L}}|G'|$ for all $|S_{i}|\leqslant k$, 
implies those correlations have nonlocality depth $k+1$. As before, $T$ is a common term appearing in all the negative parts $\mathcal{I}_{-}^{g}$ under the decomposition \eqref{eq:pm_decomposition}.
\end{thm}
\begin{proof}
This theorem can be proven in the same way as Theorem \ref{thm:GMNL_ineq} and so we do not repeat it here. 
\end{proof}

Let us now illustrate Theorem \ref{thm3} with a couple of examples.

\textbf{Example 3.} First, let us consider the expression in Eq. (\ref{I1}), that is, 
\begin{equation}
 \sum_{j=2}^n I^{1,j}_{\boldsymbol{0}|\boldsymbol{0}}   
\end{equation}
as well as an $L_{k}$-partition of the set $\mathcal{A}$. Let us then denote by $S_{1}$ the set such that $A_{1}\in S_{1}$. Then, we can bound each term $I^{1,j}_{\boldsymbol{0}|\boldsymbol{0}}$ in a similarly way as in the proof of inequality \eqref{improved00}: for all $A_{j}\notin S_{1}$ we have $I^{1,j}_{\boldsymbol{0}|\boldsymbol{0}}\leq 0$, while for $A_{i}\in S_{1}$ we have $I^{1,i}_{\boldsymbol{0}|\boldsymbol{0}}\leqslant p(\mathbf{0}|\mathbf{0})-p(1\boldsymbol{0}_{n-1}|1\boldsymbol{0}_{n-1})$. Since the size of $S_{1}$ is upper-bounded by $k$, the following inequality holds true that for all $k$-producible $p(\boldsymbol{a}|\boldsymbol{x})$:
\begin{equation}
\sum_{j=2}^n I^{1,j}_{\boldsymbol{0}|\boldsymbol{0}}\leqslant (k-1)[p(\mathbf{0}|\mathbf{0}) -p(1\mathbf{0}_{n-1}|1\mathbf{0}_{n-1})].
\end{equation}
Violation of this inequality by some $\vec{p}$ implies that the nonlocality depth of this correlation is at least $k+1$.

\textbf{Example 4.} Interestingly, a similar idea can be applied to Ineq. \eqref{Isym} to formulate another inequality capable of detecting nonlocality depth. In this case, however, the derivation of the maximal value over the $k$-producible correlations is slightly more difficult. 

The task is to find the $L_k$ partition that maximizes the number of $I^{i,j}_{\boldsymbol{0}|\boldsymbol{0}}$ terms between parties from the same subset, which can be upper-bounded by $I_{\mathbf{0}|\mathbf{0}}^{i,j}\leqslant p(\mathbf{0}|\mathbf{0})$. Let us consider an $L_{k}$ partition for which there exist two sets $S_{l_{1}},S_{l_{2}}$ such that $|S_{l_{1}}|\leqslant |S_{l_{2}}|<k$. It is easy to show that a different $L_{k}$ partition for which $S_{l_{2}}'=S_{l_{2}} \cup \{s\}$, $S_{l_{1}}'=S_{l_{1}}\setminus \{s\}$, where $s\in S_{l_{1}}$ and $S_{l}'=S_{l}$ for any other $l$, gives a higher number of such $I^{i,j}_{\boldsymbol{0}|\boldsymbol{0}}$ terms. This implies that the optimal strategy for maximizing this number is to take the partition for which $|S_{i}|=k$ for all $i\in \{1,\ldots,L-1\}$ and $|S_{L}|=n\mod k$, where $L=\lceil n/k\rceil$. We thus obtain the following inequality
\begin{equation}
    \sum_i\sum_{j>i}I^{i,j}_{\boldsymbol{0}|\boldsymbol{0}}\leqslant\left[\left\lfloor\frac{n}{k}\right\rfloor\binom{k}{2}+(n_k)\binom{n_k}{2}\right]\pz,
\end{equation}
where $n_k\equiv n\mod k$, violation of which detects nonlocality depth $k+1$.

\subsection{A possible generalization to more outcomes}
\label{qutrits}

In this section we provide a possible way of generalizing the above inequalities detecting GMNL to Bell scenarios involving more outcomes, however, concentrating on the simplest case of $d=3$.
The most natural candidate for the seed inequality is the CGLMP Bell inequality \cite{Collins_2002}, 
which in the general $d$-outcome can be written in a form highly resembling the one in Eq. \eqref{CH} as \cite{PhysRevLett.100.120406,Zohren_2010},
\begin{equation}
\begin{aligned}
P(A_0&<B_0)-P(A_0<B_1)\\
&-P(B_1<A_1)-P(A_1<B_0)\leqslant 0,
\end{aligned}
\end{equation}
where 
\begin{equation}
P(A_i<B_j)=\sum_{\substack{a,b=0\\a<b}}^{d-1}p(ab|ij)    \qquad (i,j=0,1).
\end{equation}
In the particular case of $d=3$, after relabeling the outcomes of $B_0$ such that $0\leftrightarrow 2$, the above expands as
\begin{eqnarray}\label{cglmp1} 
J_{3}&:=&p(01|00)+p(00|00)+p(10|00)\nonumber\\
&&-p(01|01)-p(02|01)-p(12|01)\nonumber\\
&&-p(10|11)-p(20|11)-p(21|11)\nonumber\\
&&-p(01|10)-p(00|10)-p(10|10)\leqslant 0.
\end{eqnarray}
We work in the tripartite scenario and hence the $\mathcal{I}^{g}$ expressions will be denoted more conveniently by $\mathcal{I}^{AB}$, $\mathcal{I}^{AC}$ and $\mathcal{I}^{BC}$. To satisfy the condition that all $\mathcal{I}^{g}$'s entering the expression \eqref{eq4} have the same positive part we will add two expressions together, one lifted by $1|0$ and another by $0|0$. Yet, the expression \eqref{cglmp1} has a different number of positive terms with outcomes $1$ and $0$ on a given position. In order to equalize these terms we will add $\tilde{J}_3$ which is obtained from $J_3$ by permuting the outcomes $0$ and $1$, that is,
\begin{eqnarray}\label{cglmp2}    
\tilde{J}_3&:=&p(10|00)+p(11|00)+p(01|00)\nonumber\\
&&-p(10|01)-p(12|01)-p(02|01)\nonumber\\
&&-p(01|11)-p(21|11)-p(20|11)\nonumber\\
&&-p(10|10)-p(11|10)-p(01|10).
\end{eqnarray}
We then consider the following expression
\begin{equation}
\mathcal{I}^{i,j}=J_3^{i,j}({1|0})+\tilde{J}_3^{i,j}({0|0}),
\end{equation}
where $J_3(1|0)$ is the expression $J_3$ lifted to the third party with measurement $0$ and outcome $1$, whereas $\tilde{J}_3(0|0)$
is obtained by lifting $\tilde{J}_3$ to the third party with measurement choice and outcome set to 0. 
In this way we achieved that the positive terms in $\mathcal{I}^{AB}$, $\mathcal{I}^{BC}$ and $\mathcal{I}^{AC}$ are the same,
\begin{eqnarray}
\mathcal{I}_{+}&:=&\mathcal{I}_{+}^{AB}=\mathcal{I}_{+}^{AC}=\mathcal{I}_{+}^{BC}\nonumber\\
&=&p(001|000)+p(010|000)+p(011|000)\nonumber\\
&&+p(100|000)+p(101|000)+p(110|000).\nonumber\\
\end{eqnarray}
Additionally $\mathcal{I}^{AB}$ and $\mathcal{I}^{AC}$ share some common negative terms
\begin{eqnarray}
T &=& p(001|100)+p(010|100)+p(011|100)\nonumber\\
&&+p(100|100)+p(101|100)+p(110|100).
\end{eqnarray}
Now, we have everything to construct Bell inequalities which analogously to inequalities \eqref{Isym} and \eqref{I1} witness GMNL in $(3,2,3)$ scenario,
\begin{equation}
\mathcal{I}^{AB}+\mathcal{I}^{AC}+\mathcal{I}^{BC}\leqslant\mathcal{I}^{+}
\end{equation}
and
\begin{equation}\label{I13}
\mathcal{I}^{AB}+\mathcal{I}^{AC}\leqslant \mathcal{I}_{+} -T.  
\end{equation}

Both these inequalities can be violated using a three-qutrit GHZ state
$\ket{\mathrm{GHZ}_{3,3}}$ [cf. Eq. \eqref{GHZ}].
Inequality \eqref{I13} was further tested on pure states drawn randomly with the Haar measure (notice that biseparable states form a set of measure zero), which are invariant under a permutation of two parties. We tested $100$ states and in each case, we managed to find measurements such that obtained correlations violated \eqref{I13}. This may suggest that the analytical results from Ref. \cite{Curchod_2019} that every pure tripartite GME state which is invariant under permutation of two the parties $B$ and $C$ is GMNL, generalize to the tripartite qutrit states. 

It is also worth pointing out that the numerically found violations are much more noise-sensitive as compared to the qubit scenario. For example, the 3-qutrit GHZ state violated inequality \eqref{I13} with maximal $q\approx0.02\%$ \eqref{noise}, whereas in the case of three-qubit GHZ more than $10\%$ of white noise could be added to still detect GMNL.

\section{Conclusions}

In this work, we showed how one can easily improve one of the classes of Bell inequalities witnessing genuine nonlocality in the multipartite scenario introduced recently in Ref. \cite{Curchod_2019}. We have shown that using the improved inequality one can prove the genuine nonlocality of every tripartite, pure, genuinely entangled state. We also demonstrated how the inequalities of Ref. \cite{Curchod_2019} can be turned into ones detecting nonlocality depth. We finally presented a possible approach to generalizing the inequalities of Ref. \cite{Curchod_2019} to the case of a higher number of outcomes. 

There is a couple of ways in which the above research line can be continued. First, generalizing our proof for tripartite genuinely entangled states, the inequalities we introduced can perhaps be used to prove that all pure, genuinely entangled states are genuinely nonlocal which is a long-standing conjecture. Next, there is a question of how to maximally increase the common negative term giving a greater noise-robustness of the inequality. Here, a couple of approaches are possible. One could, for example, rewrite the "seed" inequalities using non-signaling assumption, or perhaps lift these inequalities using different measurement settings and results, akin to what we did for $3$-outcome inequality \eqref{I13}. Another interesting question is how to generalize this approach to deliver Bell inequalities detecting GMNL in scenarios involving arbitrary number of outcomes and thus being applicable to multipartite states composed of many qudits. While we have made the first step towards achieving this goal by using the three-outcome CGLMP Bell inequality, how to obtain a general construction remains open. A related question is whether using the SATWAP Bell inequality can be advantageous in this direction \cite{SATWAP}.

Lastly, there is a possibility of adapting this technique for different scenarios. One example of such would be a network scenario, either with independent sources (see e.g. \cite{Branciard_2012}) or in the local operations and shared randomness model \cite{PhysRevLett.127.200401}.  For both of these, such a technique could prove highly desirable, as derivation of suitable Bell inequality is a notoriously difficult problem in the network scenario.

\section{Acknowledgments}

We acknowledge the support by the (Polish) National Science Center through the 
SONATA BIS Grant No. 2019/34/E/ST2/00369.



\begin{thebibliography}{33}%
\makeatletter
\providecommand \@ifxundefined [1]{%
 \@ifx{#1\undefined}
}%
\providecommand \@ifnum [1]{%
 \ifnum #1\expandafter \@firstoftwo
 \else \expandafter \@secondoftwo
 \fi
}%
\providecommand \@ifx [1]{%
 \ifx #1\expandafter \@firstoftwo
 \else \expandafter \@secondoftwo
 \fi
}%
\providecommand \natexlab [1]{#1}%
\providecommand \enquote  [1]{``#1''}%
\providecommand \bibnamefont  [1]{#1}%
\providecommand \bibfnamefont [1]{#1}%
\providecommand \citenamefont [1]{#1}%
\providecommand \href@noop [0]{\@secondoftwo}%
\providecommand \href [0]{\begingroup \@sanitize@url \@href}%
\providecommand \@href[1]{\@@startlink{#1}\@@href}%
\providecommand \@@href[1]{\endgroup#1\@@endlink}%
\providecommand \@sanitize@url [0]{\catcode `\\12\catcode `\$12\catcode
  `\&12\catcode `\#12\catcode `\^12\catcode `\_12\catcode `\%12\relax}%
\providecommand \@@startlink[1]{}%
\providecommand \@@endlink[0]{}%
\providecommand \url  [0]{\begingroup\@sanitize@url \@url }%
\providecommand \@url [1]{\endgroup\@href {#1}{\urlprefix }}%
\providecommand \urlprefix  [0]{URL }%
\providecommand \Eprint [0]{\href }%
\providecommand \doibase [0]{https://doi.org/}%
\providecommand \selectlanguage [0]{\@gobble}%
\providecommand \bibinfo  [0]{\@secondoftwo}%
\providecommand \bibfield  [0]{\@secondoftwo}%
\providecommand \translation [1]{[#1]}%
\providecommand \BibitemOpen [0]{}%
\providecommand \bibitemStop [0]{}%
\providecommand \bibitemNoStop [0]{.\EOS\space}%
\providecommand \EOS [0]{\spacefactor3000\relax}%
\providecommand \BibitemShut  [1]{\csname bibitem#1\endcsname}%
\let\auto@bib@innerbib\@empty
\bibitem [{\citenamefont {Acín}\ \emph {et~al.}(2007)\citenamefont {Acín},
  \citenamefont {Brunner}, \citenamefont {Gisin}, \citenamefont {Massar},
  \citenamefont {Pironio},\ and\ \citenamefont {Scarani}}]{Ac_n_2007}%
  \BibitemOpen
  \bibfield  {author} {\bibinfo {author} {\bibfnamefont {A.}~\bibnamefont
  {Acín}}, \bibinfo {author} {\bibfnamefont {N.}~\bibnamefont {Brunner}},
  \bibinfo {author} {\bibfnamefont {N.}~\bibnamefont {Gisin}}, \bibinfo
  {author} {\bibfnamefont {S.}~\bibnamefont {Massar}}, \bibinfo {author}
  {\bibfnamefont {S.}~\bibnamefont {Pironio}},\ and\ \bibinfo {author}
  {\bibfnamefont {V.}~\bibnamefont {Scarani}},\ }\bibfield  {title} {\bibinfo
  {title} {Device-independent security of quantum cryptography against
  collective attacks},\ }\bibfield  {journal} {\bibinfo  {journal} {Physical
  Review Letters}\ }\textbf {\bibinfo {volume} {98}},\ \href
  {https://doi.org/10.1103/physrevlett.98.230501}
  {10.1103/physrevlett.98.230501} (\bibinfo {year} {2007})\BibitemShut
  {NoStop}%
\bibitem [{\citenamefont {Pironio}\ \emph {et~al.}(2010)\citenamefont
  {Pironio}, \citenamefont {Acín}, \citenamefont {Massar}, \citenamefont
  {de~la Giroday}, \citenamefont {Matsukevich}, \citenamefont {Maunz},
  \citenamefont {Olmschenk}, \citenamefont {Hayes}, \citenamefont {Luo},
  \citenamefont {Manning},\ and\ \citenamefont {Monroe}}]{Pironio_2010}%
  \BibitemOpen
  \bibfield  {author} {\bibinfo {author} {\bibfnamefont {S.}~\bibnamefont
  {Pironio}}, \bibinfo {author} {\bibfnamefont {A.}~\bibnamefont {Acín}},
  \bibinfo {author} {\bibfnamefont {S.}~\bibnamefont {Massar}}, \bibinfo
  {author} {\bibfnamefont {A.~B.}\ \bibnamefont {de~la Giroday}}, \bibinfo
  {author} {\bibfnamefont {D.~N.}\ \bibnamefont {Matsukevich}}, \bibinfo
  {author} {\bibfnamefont {P.}~\bibnamefont {Maunz}}, \bibinfo {author}
  {\bibfnamefont {S.}~\bibnamefont {Olmschenk}}, \bibinfo {author}
  {\bibfnamefont {D.}~\bibnamefont {Hayes}}, \bibinfo {author} {\bibfnamefont
  {L.}~\bibnamefont {Luo}}, \bibinfo {author} {\bibfnamefont {T.~A.}\
  \bibnamefont {Manning}},\ and\ \bibinfo {author} {\bibfnamefont
  {C.}~\bibnamefont {Monroe}},\ }\bibfield  {title} {\bibinfo {title} {Random
  numbers certified by bell’s theorem},\ }\href
  {https://doi.org/10.1038/nature09008} {\bibfield  {journal} {\bibinfo
  {journal} {Nature}\ }\textbf {\bibinfo {volume} {464}},\ \bibinfo {pages}
  {1021–1024} (\bibinfo {year} {2010})}\BibitemShut {NoStop}%
\bibitem [{\citenamefont {Mayers}\ and\ \citenamefont
  {Yao}(2004)}]{10.5555/2011827.2011830}%
  \BibitemOpen
  \bibfield  {author} {\bibinfo {author} {\bibfnamefont {D.}~\bibnamefont
  {Mayers}}\ and\ \bibinfo {author} {\bibfnamefont {A.}~\bibnamefont {Yao}},\
  }\bibfield  {title} {\bibinfo {title} {Self testing quantum apparatus},\
  }\href@noop {} {\bibfield  {journal} {\bibinfo  {journal} {Quantum Info.
  Comput.}\ }\textbf {\bibinfo {volume} {4}},\ \bibinfo {pages} {273–286}
  (\bibinfo {year} {2004})}\BibitemShut {NoStop}%
\bibitem [{\citenamefont {Šupić}\ and\ \citenamefont
  {Bowles}(2020)}]{_upi__2020}%
  \BibitemOpen
  \bibfield  {author} {\bibinfo {author} {\bibfnamefont {I.}~\bibnamefont
  {Šupić}}\ and\ \bibinfo {author} {\bibfnamefont {J.}~\bibnamefont
  {Bowles}},\ }\bibfield  {title} {\bibinfo {title} {Self-testing of quantum
  systems: a review},\ }\href {https://doi.org/10.22331/q-2020-09-30-337}
  {\bibfield  {journal} {\bibinfo  {journal} {Quantum}\ }\textbf {\bibinfo
  {volume} {4}},\ \bibinfo {pages} {337} (\bibinfo {year} {2020})}\BibitemShut
  {NoStop}%
\bibitem [{\citenamefont {Werner}(1989)}]{PhysRevA.40.4277}%
  \BibitemOpen
  \bibfield  {author} {\bibinfo {author} {\bibfnamefont {R.~F.}\ \bibnamefont
  {Werner}},\ }\bibfield  {title} {\bibinfo {title} {Quantum states with
  einstein-podolsky-rosen correlations admitting a hidden-variable model},\
  }\href {https://doi.org/10.1103/PhysRevA.40.4277} {\bibfield  {journal}
  {\bibinfo  {journal} {Phys. Rev. A}\ }\textbf {\bibinfo {volume} {40}},\
  \bibinfo {pages} {4277} (\bibinfo {year} {1989})}\BibitemShut {NoStop}%
\bibitem [{\citenamefont {Augusiak}\ \emph {et~al.}(2015)\citenamefont
  {Augusiak}, \citenamefont {Demianowicz}, \citenamefont {Tura},\ and\
  \citenamefont {Ac\'{\i}n}}]{PhysRevLett.115.030404}%
  \BibitemOpen
  \bibfield  {author} {\bibinfo {author} {\bibfnamefont {R.}~\bibnamefont
  {Augusiak}}, \bibinfo {author} {\bibfnamefont {M.}~\bibnamefont
  {Demianowicz}}, \bibinfo {author} {\bibfnamefont {J.}~\bibnamefont {Tura}},\
  and\ \bibinfo {author} {\bibfnamefont {A.}~\bibnamefont {Ac\'{\i}n}},\
  }\bibfield  {title} {\bibinfo {title} {Entanglement and nonlocality are
  inequivalent for any number of parties},\ }\href
  {https://doi.org/10.1103/PhysRevLett.115.030404} {\bibfield  {journal}
  {\bibinfo  {journal} {Phys. Rev. Lett.}\ }\textbf {\bibinfo {volume} {115}},\
  \bibinfo {pages} {030404} (\bibinfo {year} {2015})}\BibitemShut {NoStop}%
\bibitem [{\citenamefont {Augusiak}\ \emph {et~al.}(2014)\citenamefont
  {Augusiak}, \citenamefont {Demianowicz},\ and\ \citenamefont
  {Acín}}]{Augusiak_2014}%
  \BibitemOpen
  \bibfield  {author} {\bibinfo {author} {\bibfnamefont {R.}~\bibnamefont
  {Augusiak}}, \bibinfo {author} {\bibfnamefont {M.}~\bibnamefont
  {Demianowicz}},\ and\ \bibinfo {author} {\bibfnamefont {A.}~\bibnamefont
  {Acín}},\ }\bibfield  {title} {\bibinfo {title} {Local hidden–variable
  models for entangled quantum states},\ }\href
  {https://doi.org/10.1088/1751-8113/47/42/424002} {\bibfield  {journal}
  {\bibinfo  {journal} {Journal of Physics A: Mathematical and Theoretical}\
  }\textbf {\bibinfo {volume} {47}},\ \bibinfo {pages} {424002} (\bibinfo
  {year} {2014})}\BibitemShut {NoStop}%
\bibitem [{\citenamefont {Bowles}\ \emph {et~al.}(2016)\citenamefont {Bowles},
  \citenamefont {Francfort}, \citenamefont {Fillettaz}, \citenamefont
  {Hirsch},\ and\ \citenamefont {Brunner}}]{Bowles_2016}%
  \BibitemOpen
  \bibfield  {author} {\bibinfo {author} {\bibfnamefont {J.}~\bibnamefont
  {Bowles}}, \bibinfo {author} {\bibfnamefont {J.}~\bibnamefont {Francfort}},
  \bibinfo {author} {\bibfnamefont {M.}~\bibnamefont {Fillettaz}}, \bibinfo
  {author} {\bibfnamefont {F.}~\bibnamefont {Hirsch}},\ and\ \bibinfo {author}
  {\bibfnamefont {N.}~\bibnamefont {Brunner}},\ }\bibfield  {title} {\bibinfo
  {title} {Genuinely multipartite entangled quantum states with fully local
  hidden variable models and hidden multipartite nonlocality},\ }\bibfield
  {journal} {\bibinfo  {journal} {Physical Review Letters}\ }\textbf {\bibinfo
  {volume} {116}},\ \href {https://doi.org/10.1103/physrevlett.116.130401}
  {10.1103/physrevlett.116.130401} (\bibinfo {year} {2016})\BibitemShut
  {NoStop}%
\bibitem [{\citenamefont {Gisin}(1991)}]{GISIN1991201}%
  \BibitemOpen
  \bibfield  {author} {\bibinfo {author} {\bibfnamefont {N.}~\bibnamefont
  {Gisin}},\ }\bibfield  {title} {\bibinfo {title} {Bell's inequality holds for
  all non-product states},\ }\href
  {https://doi.org/https://doi.org/10.1016/0375-9601(91)90805-I} {\bibfield
  {journal} {\bibinfo  {journal} {Physics Letters A}\ }\textbf {\bibinfo
  {volume} {154}},\ \bibinfo {pages} {201} (\bibinfo {year}
  {1991})}\BibitemShut {NoStop}%
\bibitem [{\citenamefont {Cabello}(2002)}]{Cabello_2002}%
  \BibitemOpen
  \bibfield  {author} {\bibinfo {author} {\bibfnamefont {A.}~\bibnamefont
  {Cabello}},\ }\bibfield  {title} {\bibinfo {title} {Bell’s inequality
  fornspin-<mml:math xmlns:mml="http://www.w3.org/1998/math/mathml"
  display="inline"><mml:mi>s</mml:mi></mml:math>particles},\ }\bibfield
  {journal} {\bibinfo  {journal} {Physical Review A}\ }\textbf {\bibinfo
  {volume} {65}},\ \href {https://doi.org/10.1103/physreva.65.062105}
  {10.1103/physreva.65.062105} (\bibinfo {year} {2002})\BibitemShut {NoStop}%
\bibitem [{\citenamefont {Popescu}\ and\ \citenamefont
  {Rohrlich}(1992)}]{POPESCU1992293}%
  \BibitemOpen
  \bibfield  {author} {\bibinfo {author} {\bibfnamefont {S.}~\bibnamefont
  {Popescu}}\ and\ \bibinfo {author} {\bibfnamefont {D.}~\bibnamefont
  {Rohrlich}},\ }\bibfield  {title} {\bibinfo {title} {Generic quantum
  nonlocality},\ }\href
  {https://doi.org/https://doi.org/10.1016/0375-9601(92)90711-T} {\bibfield
  {journal} {\bibinfo  {journal} {Physics Letters A}\ }\textbf {\bibinfo
  {volume} {166}},\ \bibinfo {pages} {293} (\bibinfo {year}
  {1992})}\BibitemShut {NoStop}%
\bibitem [{\citenamefont {Seevinck}\ and\ \citenamefont
  {Uffink}(2001)}]{PhysRevA.65.012107}%
  \BibitemOpen
  \bibfield  {author} {\bibinfo {author} {\bibfnamefont {M.}~\bibnamefont
  {Seevinck}}\ and\ \bibinfo {author} {\bibfnamefont {J.}~\bibnamefont
  {Uffink}},\ }\bibfield  {title} {\bibinfo {title} {Sufficient conditions for
  three-particle entanglement and their tests in recent experiments},\ }\href
  {https://doi.org/10.1103/PhysRevA.65.012107} {\bibfield  {journal} {\bibinfo
  {journal} {Phys. Rev. A}\ }\textbf {\bibinfo {volume} {65}},\ \bibinfo
  {pages} {012107} (\bibinfo {year} {2001})}\BibitemShut {NoStop}%
\bibitem [{\citenamefont {Svetlichny}(1987)}]{PhysRevD.35.3066}%
  \BibitemOpen
  \bibfield  {author} {\bibinfo {author} {\bibfnamefont {G.}~\bibnamefont
  {Svetlichny}},\ }\bibfield  {title} {\bibinfo {title} {Distinguishing
  three-body from two-body nonseparability by a bell-type inequality},\ }\href
  {https://doi.org/10.1103/PhysRevD.35.3066} {\bibfield  {journal} {\bibinfo
  {journal} {Phys. Rev. D}\ }\textbf {\bibinfo {volume} {35}},\ \bibinfo
  {pages} {3066} (\bibinfo {year} {1987})}\BibitemShut {NoStop}%
\bibitem [{\citenamefont {Bancal}\ \emph {et~al.}(2013)\citenamefont {Bancal},
  \citenamefont {Barrett}, \citenamefont {Gisin},\ and\ \citenamefont
  {Pironio}}]{PhysRevA.88.014102}%
  \BibitemOpen
  \bibfield  {author} {\bibinfo {author} {\bibfnamefont {J.-D.}\ \bibnamefont
  {Bancal}}, \bibinfo {author} {\bibfnamefont {J.}~\bibnamefont {Barrett}},
  \bibinfo {author} {\bibfnamefont {N.}~\bibnamefont {Gisin}},\ and\ \bibinfo
  {author} {\bibfnamefont {S.}~\bibnamefont {Pironio}},\ }\bibfield  {title}
  {\bibinfo {title} {Definitions of multipartite nonlocality},\ }\href
  {https://doi.org/10.1103/PhysRevA.88.014102} {\bibfield  {journal} {\bibinfo
  {journal} {Phys. Rev. A}\ }\textbf {\bibinfo {volume} {88}},\ \bibinfo
  {pages} {014102} (\bibinfo {year} {2013})}\BibitemShut {NoStop}%
\bibitem [{\citenamefont {Yu}\ and\ \citenamefont
  {Oh}(2013)}]{yu2013tripartite}%
  \BibitemOpen
  \bibfield  {author} {\bibinfo {author} {\bibfnamefont {S.}~\bibnamefont
  {Yu}}\ and\ \bibinfo {author} {\bibfnamefont {C.~H.}\ \bibnamefont {Oh}},\
  }\href@noop {} {\bibinfo {title} {Tripartite entangled pure states are
  tripartite nonlocal}} (\bibinfo {year} {2013}),\ \Eprint
  {https://arxiv.org/abs/1306.5330} {arXiv:1306.5330 [quant-ph]} \BibitemShut
  {NoStop}%
\bibitem [{\citenamefont {Chen}\ \emph {et~al.}(2014)\citenamefont {Chen},
  \citenamefont {Yu}, \citenamefont {Zhang}, \citenamefont {Lai},\ and\
  \citenamefont {Oh}}]{Chen_2014}%
  \BibitemOpen
  \bibfield  {author} {\bibinfo {author} {\bibfnamefont {Q.}~\bibnamefont
  {Chen}}, \bibinfo {author} {\bibfnamefont {S.}~\bibnamefont {Yu}}, \bibinfo
  {author} {\bibfnamefont {C.}~\bibnamefont {Zhang}}, \bibinfo {author}
  {\bibfnamefont {C.}~\bibnamefont {Lai}},\ and\ \bibinfo {author}
  {\bibfnamefont {C.~H.}\ \bibnamefont {Oh}},\ }\bibfield  {title} {\bibinfo
  {title} {Test of genuine multipartite nonlocality without inequalities},\
  }\bibfield  {journal} {\bibinfo  {journal} {Physical Review Letters}\
  }\textbf {\bibinfo {volume} {112}},\ \href
  {https://doi.org/10.1103/PhysRevLett.112.140404}
  {10.1103/PhysRevLett.112.140404} (\bibinfo {year} {2014})\BibitemShut
  {NoStop}%
\bibitem [{\citenamefont {Bancal}\ \emph {et~al.}(2009)\citenamefont {Bancal},
  \citenamefont {Branciard}, \citenamefont {Gisin},\ and\ \citenamefont
  {Pironio}}]{Bancal_2009}%
  \BibitemOpen
  \bibfield  {author} {\bibinfo {author} {\bibfnamefont {J.-D.}\ \bibnamefont
  {Bancal}}, \bibinfo {author} {\bibfnamefont {C.}~\bibnamefont {Branciard}},
  \bibinfo {author} {\bibfnamefont {N.}~\bibnamefont {Gisin}},\ and\ \bibinfo
  {author} {\bibfnamefont {S.}~\bibnamefont {Pironio}},\ }\bibfield  {title}
  {\bibinfo {title} {Quantifying multipartite nonlocality},\ }\bibfield
  {journal} {\bibinfo  {journal} {Physical Review Letters}\ }\textbf {\bibinfo
  {volume} {103}},\ \href {https://doi.org/10.1103/physrevlett.103.090503}
  {10.1103/physrevlett.103.090503} (\bibinfo {year} {2009})\BibitemShut
  {NoStop}%
\bibitem [{\citenamefont {Aolita}\ \emph {et~al.}(2012)\citenamefont {Aolita},
  \citenamefont {Gallego}, \citenamefont {Cabello},\ and\ \citenamefont
  {Ac\'{\i}n}}]{PhysRevLett.108.100401}%
  \BibitemOpen
  \bibfield  {author} {\bibinfo {author} {\bibfnamefont {L.}~\bibnamefont
  {Aolita}}, \bibinfo {author} {\bibfnamefont {R.}~\bibnamefont {Gallego}},
  \bibinfo {author} {\bibfnamefont {A.}~\bibnamefont {Cabello}},\ and\ \bibinfo
  {author} {\bibfnamefont {A.}~\bibnamefont {Ac\'{\i}n}},\ }\bibfield  {title}
  {\bibinfo {title} {Fully nonlocal, monogamous, and random genuinely
  multipartite quantum correlations},\ }\href
  {https://doi.org/10.1103/PhysRevLett.108.100401} {\bibfield  {journal}
  {\bibinfo  {journal} {Phys. Rev. Lett.}\ }\textbf {\bibinfo {volume} {108}},\
  \bibinfo {pages} {100401} (\bibinfo {year} {2012})}\BibitemShut {NoStop}%
\bibitem [{\citenamefont {Augusiak}\ \emph {et~al.}(2019)\citenamefont
  {Augusiak}, \citenamefont {Salavrakos}, \citenamefont {Tura},\ and\
  \citenamefont {Acín}}]{Augusiak_2019}%
  \BibitemOpen
  \bibfield  {author} {\bibinfo {author} {\bibfnamefont {R.}~\bibnamefont
  {Augusiak}}, \bibinfo {author} {\bibfnamefont {A.}~\bibnamefont
  {Salavrakos}}, \bibinfo {author} {\bibfnamefont {J.}~\bibnamefont {Tura}},\
  and\ \bibinfo {author} {\bibfnamefont {A.}~\bibnamefont {Acín}},\ }\bibfield
   {title} {\bibinfo {title} {Bell inequalities tailored to the
  greenberger–horne–zeilinger states of arbitrary local dimension},\ }\href
  {https://doi.org/10.1088/1367-2630/ab4d9f} {\bibfield  {journal} {\bibinfo
  {journal} {New Journal of Physics}\ }\textbf {\bibinfo {volume} {21}},\
  \bibinfo {pages} {113001} (\bibinfo {year} {2019})}\BibitemShut {NoStop}%
\bibitem [{\citenamefont {Pandit}\ \emph {et~al.}(2022)\citenamefont {Pandit},
  \citenamefont {Barasiński}, \citenamefont {Márton}, \citenamefont
  {Vértesi},\ and\ \citenamefont {Laskowski}}]{Pandit_2022}%
  \BibitemOpen
  \bibfield  {author} {\bibinfo {author} {\bibfnamefont {M.}~\bibnamefont
  {Pandit}}, \bibinfo {author} {\bibfnamefont {A.}~\bibnamefont {Barasiński}},
  \bibinfo {author} {\bibfnamefont {I.}~\bibnamefont {Márton}}, \bibinfo
  {author} {\bibfnamefont {T.}~\bibnamefont {Vértesi}},\ and\ \bibinfo
  {author} {\bibfnamefont {W.}~\bibnamefont {Laskowski}},\ }\bibfield  {title}
  {\bibinfo {title} {Optimal tests of genuine multipartite nonlocality},\
  }\href {https://doi.org/10.1088/1367-2630/aca8c8} {\bibfield  {journal}
  {\bibinfo  {journal} {New Journal of Physics}\ }\textbf {\bibinfo {volume}
  {24}},\ \bibinfo {pages} {123017} (\bibinfo {year} {2022})}\BibitemShut
  {NoStop}%
\bibitem [{\citenamefont {Curchod}\ \emph {et~al.}(2019)\citenamefont
  {Curchod}, \citenamefont {Almeida},\ and\ \citenamefont
  {Acín}}]{Curchod_2019}%
  \BibitemOpen
  \bibfield  {author} {\bibinfo {author} {\bibfnamefont {F.~J.}\ \bibnamefont
  {Curchod}}, \bibinfo {author} {\bibfnamefont {M.~L.}\ \bibnamefont
  {Almeida}},\ and\ \bibinfo {author} {\bibfnamefont {A.}~\bibnamefont
  {Acín}},\ }\bibfield  {title} {\bibinfo {title} {A versatile construction of
  bell inequalities for the multipartite scenario},\ }\href
  {https://doi.org/10.1088/1367-2630/aaff2d} {\bibfield  {journal} {\bibinfo
  {journal} {New Journal of Physics}\ }\textbf {\bibinfo {volume} {21}},\
  \bibinfo {pages} {023016} (\bibinfo {year} {2019})}\BibitemShut {NoStop}%
\bibitem [{\citenamefont {Clauser}\ \emph {et~al.}(1969)\citenamefont
  {Clauser}, \citenamefont {Horne}, \citenamefont {Shimony},\ and\
  \citenamefont {Holt}}]{PhysRevLett.23.880}%
  \BibitemOpen
  \bibfield  {author} {\bibinfo {author} {\bibfnamefont {J.~F.}\ \bibnamefont
  {Clauser}}, \bibinfo {author} {\bibfnamefont {M.~A.}\ \bibnamefont {Horne}},
  \bibinfo {author} {\bibfnamefont {A.}~\bibnamefont {Shimony}},\ and\ \bibinfo
  {author} {\bibfnamefont {R.~A.}\ \bibnamefont {Holt}},\ }\bibfield  {title}
  {\bibinfo {title} {Proposed experiment to test local hidden-variable
  theories},\ }\href {https://doi.org/10.1103/PhysRevLett.23.880} {\bibfield
  {journal} {\bibinfo  {journal} {Phys. Rev. Lett.}\ }\textbf {\bibinfo
  {volume} {23}},\ \bibinfo {pages} {880} (\bibinfo {year} {1969})}\BibitemShut
  {NoStop}%
\bibitem [{\citenamefont {Contreras-Tejada}\ \emph {et~al.}(2021)\citenamefont
  {Contreras-Tejada}, \citenamefont {Palazuelos},\ and\ \citenamefont
  {de~Vicente}}]{Contreras_Tejada_2021}%
  \BibitemOpen
  \bibfield  {author} {\bibinfo {author} {\bibfnamefont {P.}~\bibnamefont
  {Contreras-Tejada}}, \bibinfo {author} {\bibfnamefont {C.}~\bibnamefont
  {Palazuelos}},\ and\ \bibinfo {author} {\bibfnamefont {J.~I.}\ \bibnamefont
  {de~Vicente}},\ }\bibfield  {title} {\bibinfo {title} {Genuine multipartite
  nonlocality is intrinsic to quantum networks},\ }\bibfield  {journal}
  {\bibinfo  {journal} {Physical Review Letters}\ }\textbf {\bibinfo {volume}
  {126}},\ \href {https://doi.org/10.1103/physrevlett.126.040501}
  {10.1103/physrevlett.126.040501} (\bibinfo {year} {2021})\BibitemShut
  {NoStop}%
\bibitem [{\citenamefont {Collins}\ \emph {et~al.}(2002)\citenamefont
  {Collins}, \citenamefont {Gisin}, \citenamefont {Linden}, \citenamefont
  {Massar},\ and\ \citenamefont {Popescu}}]{Collins_2002}%
  \BibitemOpen
  \bibfield  {author} {\bibinfo {author} {\bibfnamefont {D.}~\bibnamefont
  {Collins}}, \bibinfo {author} {\bibfnamefont {N.}~\bibnamefont {Gisin}},
  \bibinfo {author} {\bibfnamefont {N.}~\bibnamefont {Linden}}, \bibinfo
  {author} {\bibfnamefont {S.}~\bibnamefont {Massar}},\ and\ \bibinfo {author}
  {\bibfnamefont {S.}~\bibnamefont {Popescu}},\ }\bibfield  {title} {\bibinfo
  {title} {Bell inequalities for arbitrarily high-dimensional systems},\
  }\bibfield  {journal} {\bibinfo  {journal} {Physical Review Letters}\
  }\textbf {\bibinfo {volume} {88}},\ \href
  {https://doi.org/10.1103/physrevlett.88.040404}
  {10.1103/physrevlett.88.040404} (\bibinfo {year} {2002})\BibitemShut
  {NoStop}%
\bibitem [{\citenamefont {Greenberger}\ \emph {et~al.}(1989)\citenamefont
  {Greenberger}, \citenamefont {Horne},\ and\ \citenamefont
  {Zeilinger}}]{Greenberger1989}%
  \BibitemOpen
  \bibfield  {author} {\bibinfo {author} {\bibfnamefont {D.~M.}\ \bibnamefont
  {Greenberger}}, \bibinfo {author} {\bibfnamefont {M.~A.}\ \bibnamefont
  {Horne}},\ and\ \bibinfo {author} {\bibfnamefont {A.}~\bibnamefont
  {Zeilinger}},\ }\bibinfo {title} {Going beyond bell's theorem},\ in\ \href
  {https://doi.org/10.1007/978-94-017-0849-4_10} {\emph {\bibinfo {booktitle}
  {Bell's Theorem, Quantum Theory and Conceptions of the Universe}}},\ \bibinfo
  {editor} {edited by\ \bibinfo {editor} {\bibfnamefont {M.}~\bibnamefont
  {Kafatos}}}\ (\bibinfo  {publisher} {Springer Netherlands},\ \bibinfo
  {address} {Dordrecht},\ \bibinfo {year} {1989})\ pp.\ \bibinfo {pages}
  {69--72}\BibitemShut {NoStop}%
\bibitem [{\citenamefont {Barrett}\ \emph {et~al.}(2005)\citenamefont
  {Barrett}, \citenamefont {Linden}, \citenamefont {Massar}, \citenamefont
  {Pironio}, \citenamefont {Popescu},\ and\ \citenamefont
  {Roberts}}]{Barrett_2005}%
  \BibitemOpen
  \bibfield  {author} {\bibinfo {author} {\bibfnamefont {J.}~\bibnamefont
  {Barrett}}, \bibinfo {author} {\bibfnamefont {N.}~\bibnamefont {Linden}},
  \bibinfo {author} {\bibfnamefont {S.}~\bibnamefont {Massar}}, \bibinfo
  {author} {\bibfnamefont {S.}~\bibnamefont {Pironio}}, \bibinfo {author}
  {\bibfnamefont {S.}~\bibnamefont {Popescu}},\ and\ \bibinfo {author}
  {\bibfnamefont {D.}~\bibnamefont {Roberts}},\ }\bibfield  {title} {\bibinfo
  {title} {Nonlocal correlations as an information-theoretic resource},\
  }\bibfield  {journal} {\bibinfo  {journal} {Physical Review A}\ }\textbf
  {\bibinfo {volume} {71}},\ \href {https://doi.org/10.1103/physreva.71.022101}
  {10.1103/physreva.71.022101} (\bibinfo {year} {2005})\BibitemShut {NoStop}%
\bibitem [{\citenamefont {Stachura}\ \emph {et~al.}(2024)\citenamefont
  {Stachura}, \citenamefont {Makuta},\ and\ \citenamefont
  {Augusiak}}]{stachura_2024_13343476}%
  \BibitemOpen
  \bibfield  {author} {\bibinfo {author} {\bibfnamefont {I.}~\bibnamefont
  {Stachura}}, \bibinfo {author} {\bibfnamefont {O.}~\bibnamefont {Makuta}},\
  and\ \bibinfo {author} {\bibfnamefont {R.}~\bibnamefont {Augusiak}},\ }\href
  {https://doi.org/10.5281/zenodo.13343476} {\bibinfo {title} {{Single Bell
  inequality to detect genuine nonlocality in three-qubit pure genuinely
  entangled states - Mathematica notebook and data}}} (\bibinfo {year}
  {2024})\BibitemShut {NoStop}%
\bibitem [{\citenamefont {Carteret}\ \emph {et~al.}(2000)\citenamefont
  {Carteret}, \citenamefont {Higuchi},\ and\ \citenamefont
  {Sudbery}}]{Carteret_2000}%
  \BibitemOpen
  \bibfield  {author} {\bibinfo {author} {\bibfnamefont {H.~A.}\ \bibnamefont
  {Carteret}}, \bibinfo {author} {\bibfnamefont {A.}~\bibnamefont {Higuchi}},\
  and\ \bibinfo {author} {\bibfnamefont {A.}~\bibnamefont {Sudbery}},\
  }\bibfield  {title} {\bibinfo {title} {Multipartite generalization of the
  schmidt decomposition},\ }\href {https://doi.org/10.1063/1.1319516}
  {\bibfield  {journal} {\bibinfo  {journal} {Journal of Mathematical Physics}\
  }\textbf {\bibinfo {volume} {41}},\ \bibinfo {pages} {7932–7939} (\bibinfo
  {year} {2000})}\BibitemShut {NoStop}%
\bibitem [{\citenamefont {Zohren}\ and\ \citenamefont
  {Gill}(2008)}]{PhysRevLett.100.120406}%
  \BibitemOpen
  \bibfield  {author} {\bibinfo {author} {\bibfnamefont {S.}~\bibnamefont
  {Zohren}}\ and\ \bibinfo {author} {\bibfnamefont {R.~D.}\ \bibnamefont
  {Gill}},\ }\bibfield  {title} {\bibinfo {title} {Maximal violation of the
  collins-gisin-linden-massar-popescu inequality for infinite dimensional
  states},\ }\href {https://doi.org/10.1103/PhysRevLett.100.120406} {\bibfield
  {journal} {\bibinfo  {journal} {Phys. Rev. Lett.}\ }\textbf {\bibinfo
  {volume} {100}},\ \bibinfo {pages} {120406} (\bibinfo {year}
  {2008})}\BibitemShut {NoStop}%
\bibitem [{\citenamefont {Zohren}\ \emph {et~al.}(2010)\citenamefont {Zohren},
  \citenamefont {Reska}, \citenamefont {Gill},\ and\ \citenamefont
  {Westra}}]{Zohren_2010}%
  \BibitemOpen
  \bibfield  {author} {\bibinfo {author} {\bibfnamefont {S.}~\bibnamefont
  {Zohren}}, \bibinfo {author} {\bibfnamefont {P.}~\bibnamefont {Reska}},
  \bibinfo {author} {\bibfnamefont {R.~D.}\ \bibnamefont {Gill}},\ and\
  \bibinfo {author} {\bibfnamefont {W.}~\bibnamefont {Westra}},\ }\bibfield
  {title} {\bibinfo {title} {A tight tsirelson inequality for infinitely many
  outcomes},\ }\href {https://doi.org/10.1209/0295-5075/90/10002} {\bibfield
  {journal} {\bibinfo  {journal} {EPL (Europhysics Letters)}\ }\textbf
  {\bibinfo {volume} {90}},\ \bibinfo {pages} {10002} (\bibinfo {year}
  {2010})}\BibitemShut {NoStop}%
\bibitem [{\citenamefont {Salavrakos}\ \emph {et~al.}(2017)\citenamefont
  {Salavrakos}, \citenamefont {Augusiak}, \citenamefont {Tura}, \citenamefont
  {Wittek}, \citenamefont {Ac\'{\i}n},\ and\ \citenamefont {Pironio}}]{SATWAP}%
  \BibitemOpen
  \bibfield  {author} {\bibinfo {author} {\bibfnamefont {A.}~\bibnamefont
  {Salavrakos}}, \bibinfo {author} {\bibfnamefont {R.}~\bibnamefont
  {Augusiak}}, \bibinfo {author} {\bibfnamefont {J.}~\bibnamefont {Tura}},
  \bibinfo {author} {\bibfnamefont {P.}~\bibnamefont {Wittek}}, \bibinfo
  {author} {\bibfnamefont {A.}~\bibnamefont {Ac\'{\i}n}},\ and\ \bibinfo
  {author} {\bibfnamefont {S.}~\bibnamefont {Pironio}},\ }\bibfield  {title}
  {\bibinfo {title} {Bell inequalities tailored to maximally entangled
  states},\ }\href {https://doi.org/10.1103/PhysRevLett.119.040402} {\bibfield
  {journal} {\bibinfo  {journal} {Phys. Rev. Lett.}\ }\textbf {\bibinfo
  {volume} {119}},\ \bibinfo {pages} {040402} (\bibinfo {year}
  {2017})}\BibitemShut {NoStop}%
\bibitem [{\citenamefont {Branciard}\ \emph {et~al.}(2012)\citenamefont
  {Branciard}, \citenamefont {Rosset}, \citenamefont {Gisin},\ and\
  \citenamefont {Pironio}}]{Branciard_2012}%
  \BibitemOpen
  \bibfield  {author} {\bibinfo {author} {\bibfnamefont {C.}~\bibnamefont
  {Branciard}}, \bibinfo {author} {\bibfnamefont {D.}~\bibnamefont {Rosset}},
  \bibinfo {author} {\bibfnamefont {N.}~\bibnamefont {Gisin}},\ and\ \bibinfo
  {author} {\bibfnamefont {S.}~\bibnamefont {Pironio}},\ }\bibfield  {title}
  {\bibinfo {title} {Bilocal versus nonbilocal correlations in
  entanglement-swapping experiments},\ }\bibfield  {journal} {\bibinfo
  {journal} {Physical Review A}\ }\textbf {\bibinfo {volume} {85}},\ \href
  {https://doi.org/10.1103/physreva.85.032119} {10.1103/physreva.85.032119}
  (\bibinfo {year} {2012})\BibitemShut {NoStop}%
\bibitem [{\citenamefont {Coiteux-Roy}\ \emph {et~al.}(2021)\citenamefont
  {Coiteux-Roy}, \citenamefont {Wolfe},\ and\ \citenamefont
  {Renou}}]{PhysRevLett.127.200401}%
  \BibitemOpen
  \bibfield  {author} {\bibinfo {author} {\bibfnamefont {X.}~\bibnamefont
  {Coiteux-Roy}}, \bibinfo {author} {\bibfnamefont {E.}~\bibnamefont {Wolfe}},\
  and\ \bibinfo {author} {\bibfnamefont {M.-O.}\ \bibnamefont {Renou}},\
  }\bibfield  {title} {\bibinfo {title} {No bipartite-nonlocal causal theory
  can explain nature's correlations},\ }\href
  {https://doi.org/10.1103/PhysRevLett.127.200401} {\bibfield  {journal}
  {\bibinfo  {journal} {Phys. Rev. Lett.}\ }\textbf {\bibinfo {volume} {127}},\
  \bibinfo {pages} {200401} (\bibinfo {year} {2021})}\BibitemShut {NoStop}%
\end{thebibliography}

%

\newpage
\onecolumngrid
\appendix
\section{Proof of Theorem \ref{thm}}\label{app:GME=GMNL}
In this appendix, we present a detailed proof of Thm. \ref{thm}, which we recall below
\setcounter{thm}{1}
\begin{thm}
Every tripartite, pure, genuinely multipartite state violates Ineq. \eqref{improved00}. 
\end{thm}
\begin{proof}
By the virtue of \cite[Theorem 3]{Carteret_2000}, every pure, $3$-partite, qubit state can be written as
\begin{equation}\label{eq:decomposition}
\ket{\psi}_{ABC}=a \mathrm{e}^{\mathbbm{i}\varphi}\ket{000}+b\ket{011}+c\ket{101}+d\ket{110}+e\ket{111},
\end{equation}
where $a,b,c,d,e$ are nonnegative and $a\geqslant x$ for $x\in\{b,c,d,e\}$, and by $\mathrm{e}=2.718\dots$ we denote the Euler's number. Notice, that by permuting parties, we can rearrange $b,c,d$ without affecting other terms. Therefore, without a loss of generality, we can take $b\geqslant c\geqslant d$.

Furthermore, in Ref. \cite{Curchod_2019} it was shown that the inequality \eqref{I1} is violated by all pure, genuinely multipartite entangled states that are invariant under the permutation of a pair of qubits. As the new inequality \eqref{improved00} is \replaced{finer}{tighter} than \eqref{I1}, it directly follows that it also has to be violated by such states. Therefore, we can restrict ourselves to states that are not invariant under the permutation of any two qubits, which implies that $b>c>d$.

For the case of $n=3$, Ineq. \eqref{improved00} can be rewritten as
\begin{equation}\label{ineq:I^1_3}
p(000|000)-p(100|100)-p(010|010)-p(001|001)-p(000|110)-p(000|101)\leqslant 0.
\end{equation}
Therefore, to complete the proof we have to show that for all parameters $a,b,c,d,e$ satisfying 
\begin{equation}\label{eq:conditions}
a\geqslant b>c>d\geqslant 0,\qquad a\geqslant e\geqslant 0,\qquad a^{2}+b^{2}+c^{2}+d^{2}+e^{2}=1,
\end{equation} 
there exist measurements which performed on the state \eqref{eq:decomposition} lead to the violation of Ineq. \eqref{ineq:I^1_3}.

To this end, we assume that each party performs projective measurements which we denote by $M_{i|j}^{(k)}=\Ke{M_{i|j}^{(k)}}\!\Br{M_{i|j}^{(k)}}$, where $i$ is the outcome of the measurement $j$ performed by the party $A_{k}$, and we take the states $\Ke{M_{i|j}^{(k)}}$ to be
\begin{equation}
\begin{aligned}
\Ke{M_{0|0}^{(1)}}&=\cos{\alpha}\ket{0}+\sin{\alpha}\ket{1}, &\hspace{1.5cm} \Ke{M_{0|1}^{(1)}}&=\frac{1}{\eta_{1}} \left(-c^2\sin\alpha\ket{0}+a^2\cos\alpha\ket{1}\right),\\
\Ke{M_{0|0}^{(2)}}&=\ket{0}, & \Ke{M_{0|1}^{(2)}}&=\frac{1}{\eta_{2}} \left(c\sin\alpha\ket{0}+\left(b\cos\alpha+e\sin\alpha\right)\ket{1}\right),\\
\Ke{M_{0|0}^{(3)}}&=\ket{1}, & \Ke{M_{0|1}^{(3)}}&=\frac{1}{\eta_{3}} \left(\mathrm{e}^{\mathbbm{i}\phi} a\cos\alpha\ket{0}+c\sin\alpha\ket{1}\right),
\end{aligned}
\end{equation}
where $\alpha\in [0,\pi/2 ]$ is a free parameter, and $\eta_{k}$ for $k=1,2,3$ are normalization constants
\begin{equation}\label{eq:eta}
\eta_{1}^{2}= c^4\sin^2\alpha+a^4\cos^2\alpha,\qquad\eta_{2}^{2}= c^{2}\sin^{2}{\alpha}+\left(b\cos{\alpha}+e\sin{\alpha}\right)^{2},\qquad \eta_{3}^{2}=a^{2}\cos^{2}\alpha+c^{2}\sin^{2}\alpha.
\end{equation}
Notice, that measurement operators for outcome $1$ can be trivially derived from the relation $M_{0|j}^{(k)}+M_{1|j}^{(k)}=\mathbb{1}$ for all $j,k$. It is easy to check, that for this particular choice of measurements, we have
\begin{equation}
p(001|001)=p(010|010)=p(000|101)=0,
\end{equation}
while the remaining probabilities in Ineq. \eqref{ineq:I^1_3} are given by
\begin{equation}\label{eq: probabilities_sub}
\begin{gathered}
p(000|000) = c^2\sin^2\alpha,\qquad p(100|100) = \frac{1}{\eta_{1}^{2}}c^6\sin^2\alpha,\\
p(000|110) =\frac{1}{\eta_{1}^{2}\eta_{2}^{2}}\left[a^{2}c^2 \sin\alpha\cos\alpha+\left(a^{2}e\cos\alpha -bc^2\sin\alpha\right) \left(b\cos\alpha+e\sin\alpha\right)\right]^{2}.
\end{gathered}
\end{equation}
And so, after substituting the above probabilities to Ineq. \eqref{ineq:I^1_3} we get
\begin{equation}
c^2\sin^2\alpha- \frac{1}{\eta_{1}^{2}}c^6\sin^2\alpha -\frac{1}{\eta_{1}^{2}\eta_{2}^{2}}\left[a^{2}c^2 \sin\alpha\cos\alpha+\left(a^{2}e\cos\alpha -bc^2\sin\alpha\right) \left(b\cos\alpha+e\sin\alpha\right)\right]^{2}\leqslant 0.
\end{equation}

We need to consider two separate cases to prove the violation of this inequality. First, let us assume that $e=0$. Under this assumption we can take $\alpha=\pi/4$ which yields
\begin{equation}\label{ineq:proof_substitution}
c^2\frac{1}{2}- \frac{1}{2\eta_{1}^{2}}c^6\ -\frac{1}{4\eta_{1}^{2}\eta_{2}^{2}}\left(a^{2}c^2 -b^{2}c^2\right)^{2}\leqslant 0.
\end{equation}
Next, we multiply the above by $4\eta_{1}^{2}\eta_{2}^{2}$, and we use Eq. \eqref{eq:eta} to express $\eta_{1}$ and $\eta_{2}$ in terms of the coefficients $a,b,c,e,\alpha$, which gives us
\begin{equation}\label{ineq:proof_e=0}
\frac{1}{2} (b^{2}+c^{2})(a^{2}-c^{2})(a^{2}+c^{2}) -c^{2}\left(a^{2} -b^{2}\right)^{2}\leqslant 0.
\end{equation}
However, from Ineq. \eqref{eq:conditions} it follows that
\begin{equation}
\frac{1}{2}(b^2+c^2)>c^2\quad\textrm{and}\quad a^2-c^2>a^2-b^2\quad\textrm{and}\quad a^2+a^2>a^2-b^2,
\end{equation}
implying that the left-hand side of Ineq. \eqref{ineq:proof_e=0} is always positive, which implies the violation of Ineq. \eqref{ineq:I^1_3}.

For the case $e>0$ let us first focus on the first two terms of Ineq. \eqref{ineq:proof_substitution}
\begin{equation}
c^2\sin^2\alpha- \frac{1}{\eta_{1}^{2}}c^6\sin^2\alpha.
\end{equation}
It is easy to see that the above is positive for all $\alpha \in ] 0,\pi/2 [$, and so to show the violation of Ineq. \eqref{ineq:proof_substitution} it is sufficient to show that there exists $\alpha \in ] 0,\pi/2 [$ for which the third term from Ineq. \eqref{ineq:proof_substitution} equals $0$
\begin{equation}
a^{2}c^2 \sin\alpha\cos\alpha+\left(a^{2}e\cos\alpha -bc^2\sin\alpha\right) \left(b\cos\alpha+e\sin\alpha\right) =0.
\end{equation}
Observe, that the left-hand side is strictly positive for $\alpha=0$ and strictly negative for $\alpha=\pi/2$, and since the function is continuous for all $\alpha \in [0,\pi/2]$, there has to exist $\alpha \in ] 0,\pi/2 [$ for which the above equation holds true. Therefore, Ineq. \eqref{ineq:I^1_3} is violated. 

Lastly, let us note that this proof implies that all states of the form \eqref{eq:decomposition} that are not invariant under the permutation of a pair of qubits are GMNL, and so also GME. And since the GME states that are invariant under such operations were proven to be GMNL in \cite{Curchod_2019}, we have that all pure GME states are GMNL.
\end{proof}

\end{document}